\documentclass[onecolumn,notitlepage,nofootinbib,floatfix,superscriptaddress]{revtex4-2}
\usepackage{amsmath, amsfonts, amsthm, amssymb, bbm}
\usepackage[margin=1in]{geometry}
\usepackage{xspace}
\usepackage{algorithm,algpseudocode}
\usepackage{graphicx,xcolor,tikz}
\usepackage{hyperref}
\usepackage{physics}
\usepackage{qcircuit}

\newcommand{\bb}{\mathbb}
\newcommand{\mc}{\mathcal}

\DeclareMathOperator{\poly}{poly}
\newcommand{\cc}{\texttt{ColorCodeZ}\xspace}
\newcommand{\surfc}{\texttt{SurfaceCodeXYZ}\xspace}
\newcommand{\tcnot}{\texttt{tCNOTZ}\xspace}
\newcommand{\tdm}{\texttt{3DM}\xspace}
\newcommand{\tsat}{\texttt{3-SAT}\xspace}

\algblockdefx[ForEach]{ForEach}{EndForEach}[1]{\textbf{for each} #1 \textbf{do}}{\textbf{end for each}}

\newtheorem{theorem}{Theorem}
\newtheorem{lemma}[theorem]{Lemma}
\newtheorem{assumption}[theorem]{Assumption}

\begin{document}

\title{The color code, the surface code, and the transversal CNOT:\\
NP-hardness of minimum-weight decoding}

\author{Shouzhen Gu}
\affiliation{Yale Quantum Institute \& Department of Applied Physics, Yale University, New Haven, CT, USA}
\author{Lily Wang}
\affiliation{CSE Division, University of Michigan, Ann Arbor, MI, USA}
\author{Aleksander Kubica}
\affiliation{Yale Quantum Institute \& Department of Applied Physics, Yale University, New Haven, CT, USA}

\begin{abstract}
The decoding problem is a ubiquitous algorithmic task in fault-tolerant quantum computing, and solving it efficiently is essential for scalable quantum computing.
Here, we prove that minimum-weight decoding is NP-hard in three quintessential settings:
(i) the color code with Pauli $Z$ errors,
(ii) the surface code with Pauli $X$, $Y$ and $Z$ errors,
and (iii) the surface code with a transversal CNOT gate, Pauli $Z$ and measurement bit-flip errors.
Our results show that computational intractability already arises in basic and practically relevant decoding problems central to both quantum memories and logical circuit implementations, highlighting a sharp computational complexity separation between minimum-weight decoding and its approximate realizations.
\end{abstract}

\maketitle

\section{Introduction}

The decoding problem is one of the central algorithmic tasks in quantum error correction (QEC)~\cite{Shor1995, Steane1996} and fault-tolerant quantum computation~\cite{Shor1996, Aharonov1997,Knill1998, Preskill1998}.
It can be formulated as follows: given an error syndrome describing which errors may have occurred in the system, find a recovery operator that mitigates their detrimental effects on the encoded information.
In stabilizer codes~\cite{Gottesman1996}, the error syndrome is obtained by measuring parity checks, which form a set of commuting Pauli operators that define the code space.
Two canonical scenarios for the decoding problem are the code-capacity and phenomenological noise models\footnote{We remark that the circuit noise model introduces further details by specifying syndrome extraction circuits; however, for quantum low-density parity-check codes~\cite{Breuckmann2021}, it can be captured by the phenomenological noise model (possibly with small-weight correlated errors).
That is why we do not consider it as giving rise to a qualitatively different decoding scenario.
}, in both of which qubits are affected by Pauli noise.
In the former, parity-check measurements are assumed to be perfect, while in the latter parity-check outcomes may be corrupted, typically by independent bit-flip noise.
Thus, in the phenomenological setting, reliable QEC usually requires repeated rounds of syndrome extraction in order to gain confidence in the parity-check outcomes, unless the QEC protocol exhibits the single-shot property~\cite{Bombin2015,Campbell2019,Kubica2022,Gu2024} or exploits the paradigm of algorithmic fault tolerance~\cite{Zhou2025, cain2025fastcorrelateddecoding}.

Efficient decoding algorithms are essential for scalable quantum computation---a fault-tolerant quantum computer continuously generates syndrome data and the corresponding classical processing must keep pace; otherwise, one encounters the backlog problem, in which classical post-processing becomes the bottleneck~\cite{Terhal2015}.
A decoding algorithm takes as input the measured syndrome together with a decoding hypergraph, whose vertices and hyperedges correspond to, respectively, the syndrome and elementary errors.
In the code-capacity setting, these elementary errors are typically single-qubit Pauli errors, whereas in the phenomenological setting, they also include measurement bit-flip errors affecting parity-check outcomes.
The decoding algorithm may additionally use prior information about the underlying error distribution.

A particularly important class of decoding algorithms is minimum-weight decoders, which seek a recovery consistent with the measured syndrome that has minimum weight.
This paradigm is motivated by the standard assumption of independent and identically distributed noise, for which the most-likely error usually corresponds to the minimum-weight error.
Since stabilizer codes are degenerate, i.e., different errors can have the same syndrome, minimum-weight decoders are not generally the same as the maximum-likelihood decoder, which should instead identify the most likely equivalence class of errors.
Nevertheless, minimum-weight decoders remain among the most successful and widely used decoding algorithms.
In particular, minimum-weight decoders are regarded as the gold standard for the surface code~\cite{Dennis2002, Higgott2023, Wu2023}; they also constitute a core ingredient in competitive decoders for other topological quantum codes, including the color code~\cite{Delfosse2014,Chamberland2020,Sahay2022,Kubica2023CCrestrictiondecoder} and more general two-dimensional translationally-invariant codes~\cite{tan2026generalizedmatchingdecoders2d,sahay2026matchingdecoderbivariatebicycle}.

In this work, we prove that minimum-weight decoding is NP-hard in three fundamental QEC settings: (i) the color code, (ii) the surface code and (iii) the surface code with a transversal CNOT gate.
Concretely, for (i) and (ii) we consider the code-capacity setting with, respectively, Pauli $Z$ errors and Pauli $X$, $Y$ and $Z$ errors, while for (iii) we consider the phenomenological noise with Pauli $Z$ and measurement bit-flip errors.
We focus on either finding a minimum-weight error or interpreting its logical effect.
In each case, we establish NP-hardness via a reduction from the three-dimensional matching problem, which is NP-complete~\cite{Karp1972,GareyJohnson90NPcompleteness}; see Fig.~\ref{fig:reductionoverview} for an overview of the reduction.
Our results show that computational intractability already appears in some of the most basic and practically relevant decoding problems in QEC.
At the same time, they sharply distinguish minimum-weight decoding from its approximate realizations---in all three settings, there exist efficient decoders that output a recovery whose weight is within a constant factor (two or three) of the minimum-weight recovery; see Appendix~\ref{app_minwt} for details.
Finally, we note that Walters and Turner~\cite{walters2026CCNPhard} have very recently established NP-hardness for case (i) via a reduction from the \tsat  problem.
Our work goes beyond their setting, extending similar arguments to cases (ii) and (iii) for the prototypical and the most studied QEC code, the surface code.

\section{Preliminaries}
\label{sec:prelim}

We begin with some basic complexity theory definitions.
The complexity class NP consists of decision problems where the YES instances have proofs that can be verified in polynomial time. We say that a problem $A$ is NP-hard if there is a polynomial-time reduction from any problem in NP to $A$; that is, if we can solve $A$, then we can solve any NP problem with polynomial time overhead. A problem is NP-complete if it is NP-hard and in NP.

The 3-dimensional matching problem (\tdm) is among Karp's 1972 foundational list of 21 NP-complete problems \cite{Karp1972,GareyJohnson90NPcompleteness}. In an instance of \tdm, we are given three finite sets $A$, $B$, $C$ of equal size $r$ and a set of hyperedges $T \subseteq A \times B \times C$. A subset $M \subseteq T$ is a matching if for any distinct $(a_1, b_1, c_1), (a_2, b_2, c_2) \in M$, we have $a_1 \neq a_2$, $b_1 \neq b_2$, and $c_1 \neq c_2$. The decision problem is to determine whether there exists a perfect matching, i.e., a matching of size $r$ that covers all elements. The problem size is taken to be $s=|T|$.

\begin{quote}
\textbf{3-DIMENSIONAL MATCHING (\tdm)}\\
\textbf{Instance}: A size-$s$ set $T\subseteq A\times B\times C$, where $A$, $B$, and $C$ are sets of the same cardinality $r$.\\
\textbf{Question}: 
Does $T$ contain a perfect matching, i.e., a subset $M\subseteq T$ such that $|M|=r$ and no two elements of $M$ agree in any coordinate?
\end{quote}

We next provide some background on QEC codes and the decoding problem.
A stabilizer code on $n$ qubits is defined as the $+1$-eigenspace of an Abelian subgroup $\mathcal S$ of the $n$-qubit Pauli group $\mathcal{P}_n$, with $-I \notin \mathcal S$. Since the codespace is an eigenspace, any Pauli error which anticommutes with some check in $\mathcal S$ removes the state from the codespace and can be detected. In order to correct errors, we first measure the stabilizer checks of $\mathcal S$ to obtain the \emph{syndrome} of $\pm 1$ outcomes. The syndrome can then be used to predict the most likely error and recover the original code state.

In the minimum-weight decoding problem, we are given as input a binary matrix $H\in \bb F_2^{m\times n}$, a syndrome $\sigma\in \bb F_2^m$, and a weight function on vectors in $\bb F_2^n$.
The task is to find an error $e\in \bb F_2^n$ of lowest weight such that $He=\sigma$ (if it exists).
Here, we consider the weight function that gives the size of the error's support and take $n$ to be the size of the decoding problem.
The columns of $H$ are typically referred to as variables and the rows as constraints, with the rows corresponding to 0 or 1 entries in $\sigma$ as satisfied or unsatisfied, respectively.
When decoding QEC codes, $H$ is called either the parity-check matrix in the code-capacity setting or the detector check matrix~\cite{Higgott2025sparseblossom} in the phenomenological noise setting. We will consider the following three scenarios.
\begin{enumerate}
    \item \emph{Color code under $Z$ noise.} We consider the triangular color code~\cite{bombin2006} defined on a hexagonal lattice of linear size $L$; see Fig.~\ref{fig:specialwiregadget}(a)(b) for illustration. In the dual lattice description, qubits are placed on the triangles of the triangular lattice.
    $X$ and $Z$ checks are associated with the vertices, and they are supported on the qubits on the adjacent triangles.
    In the decoding problem, the variables correspond to whether each qubit has suffered a $Z$ error, and the $X$ checks are the constraints. Hence, $H$ is the incidence matrix between the triangles and vertices on the dual lattice. The weight of an error is the Hamming weight, i.e., the number of corrupted qubits.
    \begin{quote}
    \textbf{MIN-WEIGHT DECODING FOR THE COLOR CODE (\cc)}\\
    \textbf{Instance}: The $n$-qubit triangular color code on a hexagonal lattice and an $X$ syndrome $\sigma_X$.\\
    \textbf{Question}: Find a Pauli $Z$ error $E_Z$ of minimal weight with the syndrome $\sigma_X$.
    \end{quote}
    \item \emph{Surface code under Pauli noise.} We consider the surface code~\cite{Kitaev2003,Dennis2002} defined on a square lattice, where the qubits, $X$ checks, and $Z$ checks are placed on the edges, vertices, and faces of the lattice, respectively. The support of each $X$ or $Z$ check is the set of qubits on the edges incident to the vertex or face of the check.
    In the decoding problem, there are two variables for each qubit: the first determining if it has suffered an $X$ error and the second determining if it has suffered a $Z$ error (with both being 1 if it has suffered a $Y$ error). The constraints are the $X$ and $Z$ checks. Then, $H$ is a block diagonal matrix consisting of the edge-vertex and edge-face incidence matrices.
    However, we cannot decode its two blocks independently, as the weight function returns the error's support.
    \begin{quote}
    \textbf{MIN-WEIGHT DECODING FOR THE SURFACE CODE (\surfc)}\\
    \textbf{Instance}: The $n/2$-qubit surface code on a square lattice and a syndrome $\sigma$.\\
    \textbf{Question}: Find a Pauli error $E$ of minimal weight with the syndrome $\sigma$.
    \end{quote}
    \item \emph{Transversal CNOT in the surface code with the phenomenological noise.} In this setting~\cite{Beverland21costofuniversality,Sahay25tCNOT,wan2025iterativetransversalcnotdecoder,Cain2024}, there are two square-lattice surface code patches, and we perform repeated stabilizer measurements at half-integer times $t=t_i-\frac{1}{2}, \dots, t_{G}-\frac{1}{2}, t_{G}+\frac{1}{2}, \dots, t_f+\frac{1}{2}$. At $t=t_{G}$, a transversal logical CNOT gate is applied from the first surface code to the second. The constraints are the $X$ \emph{detectors}, which are products of consecutive measurement outcomes $s^a_{t-1/2}s^a_{t+1/2}$ of the same $X$ stabilizer $s$ on the $a$-th surface code at times $t-\frac{1}{2}$ and $t+\frac{1}{2}$. However, the detectors $s^1_{t_{G}-1/2}s^1_{t_G+1/2}$ are replaced by $s^1_{t_{G}-1/2}s^1_{t_G+1/2}s^2_{t_G-1/2}$ due to the transversal CNOT gate.
    There is one variable for each qubit at each integer time step from $t_i$ to $t_f$ corresponding to whether that qubit suffered a $Z$ error at that time. For concreteness, we assume that at $t=t_{G}$ the error happens after the CNOT gate is applied.
    Additionally, there is a variable for each measurement outcome from $t_i+\frac{1}{2}$ to $t_f-\frac{1}{2}$ corresponding to an incorrect measurement. 
    For simplicity, we assume that the first and last measurement outcomes are ideal.
    Then, $H$ comprises the edge-vertex incidence matrices of two 3D cubic lattices, with the edges of the second lattice at time $t=t_G-\frac{1}{2}$ replaced by hyperedges that also connect to the vertices at the same spatial locations of the first lattice at time $t_G$.
    The weight of an error is the Hamming weight.
    \begin{quote}
    \textbf{MIN-WEIGHT DECODING FOR THE TRANSVERSAL CNOT GATE (\tcnot)}\\
    \textbf{Instance}: Two surface codes on $L\times L$ square lattices with $2L$ measurement rounds and a transversal logical CNOT gate applied after $L$ rounds, a set of triggered $X$ detectors $\sigma$.\\
    \textbf{Question}: Find a minimum-weight error $E$, consisting of $Z$ qubit errors and measurement bit-flip errors, with the syndrome $\sigma$.
    \end{quote}
\end{enumerate}

In the setting of QEC codes, we refer to unsatisfied constraints as \emph{defects}. To make the connection with \tdm, we label defects with three types $A$, $B$, $C$ in each of the three decoding scenarios, such that the number of defects of each type has the same parity for any error.
For the color code, the vertices of the triangular lattice can be colored red, green, and blue, such that no two adjacent vertices have the same color. This means a single error causes one defect of each color, which we associate with the three defect types.
In the surface code, we consider a violated $Z$ ($X$) parity check as an $A$ ($C$) defect, also called an $X$ ($Z$) defect.
Note that the defect type is labeled by the type of Pauli error that causes it, not the type of stabilizer that is violated. If a $Z$ defect is on a vertex that is adjacent to a face containing an $X$ defect, we combine the two defects and refer to them as a $B$ or $Y$ defect. This operation maintains that the parity of the number of defects of each type is the same.
In the surface code with a transversal CNOT, we label all defects on the first surface code as $B$ type, irrespective of their time coordinate.
All defects on the second surface code at times $t\ge t_G$ are labeled $C$, and those at times $t\le t_G-1$ are labeled $A$. A weight-one error creates two defects of the same type, except if it is a measurement error on the second surface code at time $t_G-\frac{1}{2}$, which creates one defect of each type. Thus, the number of defects of each type has the same parity.

In a decoding problem, we can create a \emph{Tanner graph} $\mc G$, which is a bipartite graph with the vertex set $V\sqcup W$ that captures the connectivity between errors and defects. Specifically, $V$ is the set of \emph{error locations}---the qubits in \cc and \surfc, and the qubits and measurements at each time in \tcnot---and $W$ is the set of constraints. We connect $v\in V$ and $w\in W$ if there is an error supported on $v$ that violates the constraint $w$, and we say that $v$ and $w$ are adjacent.
For subsets of $E$, we say that two error locations are adjacent if they are next nearest neighbors in $\mc G$.
A \emph{connected error component} of $E$ is the restriction of $E$ to a connected set under this metric. We may refer to the $k$-neighborhood of a set of defects $g$ as those error locations with distance at most $2k-1$ from $g$.

Finally, we note that the minimum-weight decoding problem in \cc, \surfc, and \tcnot is equivalent to finding the most likely error for a given syndrome when the error is sampled from independent and identically distributed phase-flip noise, depolarizing noise, and phase-flip with measurement bit-flip noise, respectively.

\section{Reducing \tdm to the minimum-weight decoding problems}
\label{sec:commonproof}
In this section, we prove the following theorem.

\begin{theorem}
    \label{thm:main}
    The following three decoding problems are NP-hard.
    \begin{enumerate}
        \item \emph{\cc}: in the color code on a hexagonal lattice, finding a minimum-weight $Z$ error for a given $X$ syndrome.
        \item \emph{\surfc}: in the surface code on a square lattice, finding a minimum-weight Pauli error for a given syndrome.
        \item \emph{\tcnot}: in the setting of two square-lattice surface codes with a transversal logical CNOT gate, finding a minimum-weight Pauli $Z$ and measurement bit-flip error for a given set of triggered $X$ detectors.
    \end{enumerate}
\end{theorem}

We prove Theorem~\ref{thm:main} by showing that the decision problem of determining the existence of an error of weight at most $w$ that satisfies a given syndrome is NP-complete in each of the three scenarios. The reduction follows because if we can find the minimum-weight error, we can check if its weight is at most $w$. The decision problem is in NP because a witness for a YES instance can simply be a valid low-weight error.

\begin{figure}[htpb]
    \centering
    \includegraphics[width=\textwidth,trim={1cm 2cm 1cm 2cm}, clip = True]{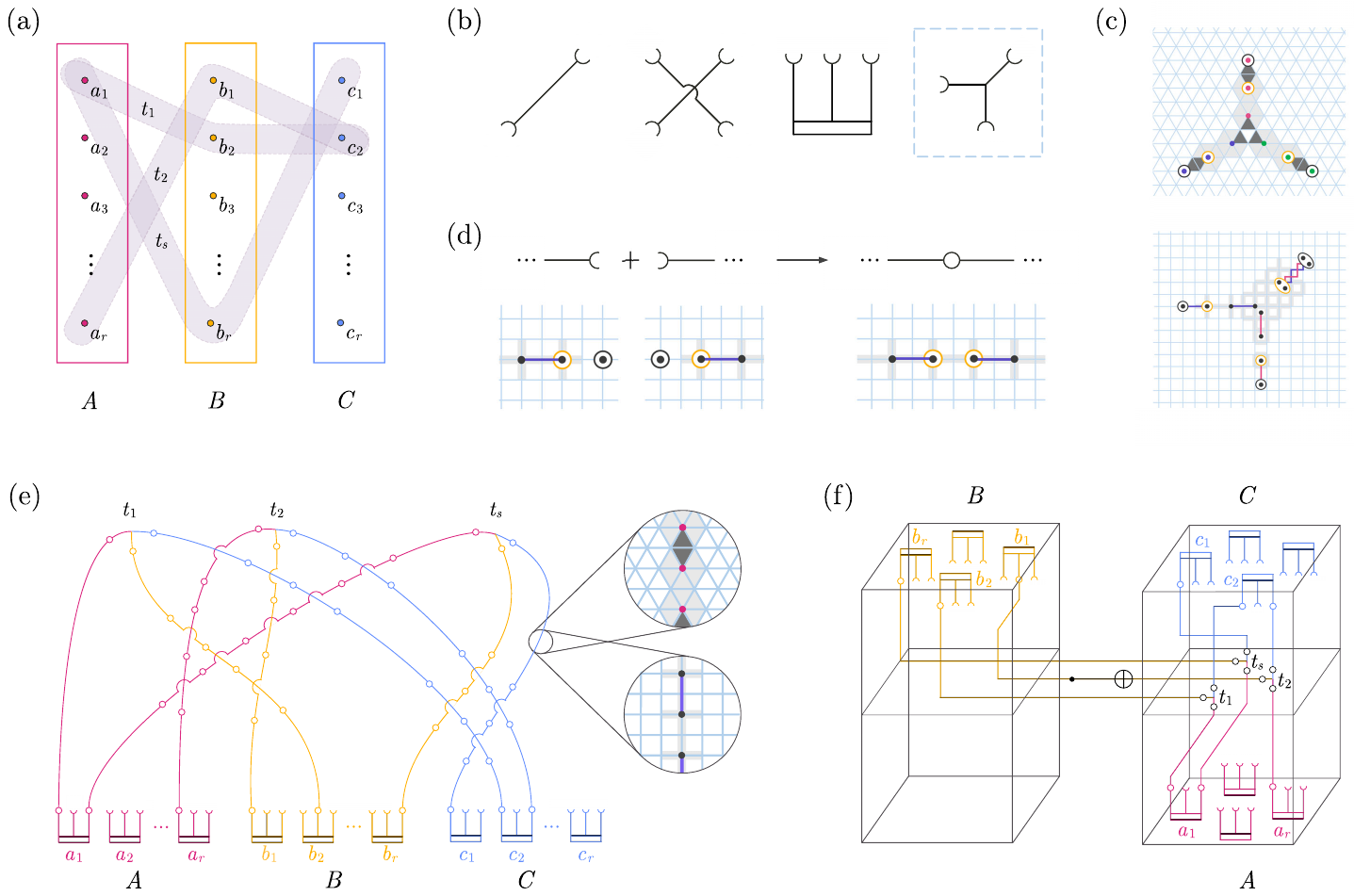}
    \caption{
    Overview of the proof of Theorem~\ref{thm:main}.
    (a) A \tdm instance, where the problem is to decide whether there exists a disjoint subset of hyperedges $t_i$ that covers all elements in $A\cup B\cup C$.
    (b) Symbols used to represent the wire, crossing, element, and splitting gadgets (from left to right).
    (c) The implementation of the splitting gadget in the color code (top) and the surface code (bottom). Dots represent defects, with the nodes circled in yellow and the partner nodes (which are defects in another gadget) circled in black. The neighborhood of the gadget is highlighted in light gray. A minimal cover in the color code with all nodes TRUE and one in the surface code with all nodes FALSE are shown.
    (d) An example of two gadgets linked together. The nodes from both gadgets have the same truth value, FALSE in this case.
    (e) Schematic for the placement of gadgets in the syndrome to be decoded. The same coarse-grained layout is used when reducing \tdm to \cc or \surfc, with differences in the implementation of the gadgets.
    (f) The placement of gadgets defining the syndrome when reducing \tdm to \tcnot.
    }
    \label{fig:reductionoverview}
\end{figure}

To show that the decision problem is NP-hard, we show a reduction from the NP-complete problem \tdm. That is, for any \tdm instance of size $s$, in each of the three scenarios, we find a decoding problem of size $\poly(s)$ such that a perfect matching exists if and only if the decoding problem has a solution of sufficiently low weight.
The main ideas in the reduction are shown in Fig.~\ref{fig:reductionoverview}.
The boundary conditions of the QEC code are irrelevant if we place the syndrome deep in the bulk of a sufficiently large lattice (e.g., linear size twice the diameter of the syndrome).

The syndrome $\sigma$ defining the decoding problem will be decomposed into various \emph{gadgets}, which are collections of defects placed in a very particular pattern; see Fig.~\ref{fig:reductionoverview}(c) for an example. The gadgets are designed so that any error with the syndrome $\sigma$ has weight at least $w$, and any valid error with weight $w$ is constrained to be in one of several specific configurations around each gadget; see Lemma~\ref{lem:errorweightlowerbound}.
The gadgets are arranged in a way dictated by the \tdm instance; see Fig.~\ref{fig:reductionoverview}(e)(f).
This way, if the errors around each gadget are in one of the allowed configurations, and the configurations around adjacent gadgets are consistent with each other, then a perfect matching for the \tdm instance can be inferred.
Conversely, a solution to the $\tdm$ instance will allow us to choose minimal error configurations around each gadget that are consistent with each other and can be combined into a valid global error of weight $w$.

Formally, a gadget $g$ comprises the following: a set of defects; a certain chosen subset, called \emph{nodes}; a set of error locations $R_g$, called the \emph{neighborhood}; and an integer $m_g$, called the \emph{error excess}.
For brevity, we often refer to the set of defects as the gadget.
The neighborhood of a defect is defined to be the 2-neighborhood in the case of a $Y$ defect in the surface code and the 1-neighborhood for all other defects. The neighborhood $R_g$ of the gadget $g$ is the union of the neighborhoods of all defects in $g$.

The gadgets have certain properties, which are marked as Assumptions throughout this section. We prove the reduction assuming these properties, deferring the details of the gadgets' construction and proofs of the required properties to Sec.~\ref{sec:gadgets}.
The placement of the gadgets and nodes will satisfy the following.

\begin{assumption}[Separation between gadgets]
    \label{ass:separation}
    The neighborhoods of two defects $p, p'$ from two different gadgets $g$, $g'$ are at distance at least two apart, with the exception that the neighborhoods of two nodes may be adjacent (but still disjoint).
    In this case, we call $p$ and $p'$ \emph{partner nodes}. Each node of a gadget will have exactly one partner node from some other gadget, which is of the same defect type.
\end{assumption}

We say that an error $E_g$ \emph{covers} the gadget $g$ if the syndrome caused by $E_g$ contains $g$.
The integer $m_g$ is chosen so that
\begin{equation}
    \label{eq:coversizebound}
    \min_{E_g: \text{ cover of $g$}} |E_g\cap R_g| = |g| + m_g.
\end{equation}
We say that any $E_g$ that achieves the minimum in Eq.~\eqref{eq:coversizebound} and is contained in $R=\bigcup_g R_g$ is a \emph{minimal cover} of $g$.
For a global error $E$ satisfying the syndrome, the \emph{connected cover} $E_g$ of a gadget $g$ is the smallest subset of $E$ that is a union of connected error components and covers $g$.

\begin{lemma}
    \label{lem:errorweightlowerbound}
    If the gadgets satisfy Assumption~\ref{ass:separation} (in particular, having disjoint neighborhoods), then any error $E$ giving syndrome $\sigma$ has weight $|E|\ge |\sigma|+m$, where $m=\sum_g m_g$. If equality holds, then the connected cover of each gadget is a minimal cover.
\end{lemma}
\begin{proof}
    We can bound the weight of $E$ as
    \begin{equation}
    \label{eq:errorwtbound}
    \textstyle
        |E| \ge \left|\bigcup_g (E_g\cap R)\right| \ge \left|\bigcup_g (E_g\cap R_g)\right| = \sum_g |E_g\cap R_g| \ge \sum_g (|g|+m_g) = |\sigma| + m,
    \end{equation}
    where $E_g$ is the connected cover of the gadget $g$. Here, we have used the disjointness of the $R_g$ in the first equality and Eq.~\eqref{eq:coversizebound} in the last inequality.

    Suppose $E$ saturates this inequality. In the bound, the first inequality becomes an equality only if $E\subseteq R$, meaning $E_g\subseteq R$ for all gadgets $g$. Then, equality holds for the last inequality exactly when $E_g$ is a minimal cover of $g$ for all $g$.
\end{proof}

While minimal covers exist for each gadget individually, there may not be a globally consistent error that has the minimum possible weight on every gadget. Our goal is to show that an error of weight $|\sigma|+m$ exists if and only if the corresponding \tdm instance has a perfect matching.

The minimal covers of the gadgets we consider cannot have syndromes at arbitrary locations outside of $g$. By inspection of the gadgets described in Sec.~\ref{sec:gadgets} and the separation between gadgets in Assumption~\ref{ass:separation}, we will arrive at the following property. 
\begin{assumption}
    \label{ass:gadgetdefectsubset}
    If $E_g$ is a minimal cover of a gadget $g$, then the syndrome caused by $E_g$ is contained in $g$ and its partner nodes.
\end{assumption}

A gadget may have multiple minimal covers. Assumption~\ref{ass:gadgetdefectsubset} motivates us to associate truth values to the nodes of a gadget in a minimal cover.
Specifically, a node is assigned TRUE if its partner node is a defect caused by the error (e.g., Fig.~\ref{fig:reductionoverview}(c)(top))
and FALSE if not (e.g., Fig.~\ref{fig:reductionoverview}(c)(bottom)).
Therefore, in a globally minimal error, gadgets that are linked together via partner nodes have the same truth value at those nodes, which is determined by if the error connects the two nodes or is disjointly supported within each gadget; see Fig.~\ref{fig:reductionoverview}(d) for an example.
A gadget will enforce that only certain combinations of truth values within its own nodes are possible.

There are four different types of gadgets: wire gadgets, crossing gadgets, element gadgets, and splitting gadgets. We now describe the key properties of these gadgets, including the truth values of their nodes in the different minimal covers.

\begin{assumption}[Properties of gadgets] The gadgets have the following properties.
    \label{ass:gadgetproperties}
    \begin{itemize}
        \item A \textbf{wire gadget} has two nodes of the same defect type, which must be both TRUE or both FALSE in a minimal cover.
        \item A \textbf{crossing gadget} has two pairs of nodes. The nodes of each pair have the same defect type and truth value and are placed on opposite sides of the gadget.
        \item An \textbf{element gadget} has $k$ nodes, all of the same defect type. In a minimal cover, exactly one of the nodes is TRUE and all others are FALSE.
        \item A \textbf{splitting gadget} has three nodes, each of different defect type. The truth values of all three nodes are equal.
    \end{itemize}
     Minimal covers with the allowed truth values exist.
     Except for the wire gadget, whose nodes may be placed at any location hosting the correct defect type, all the other gadgets have constant size.
\end{assumption}

We now explain how to connect the gadgets to each other, as illustrated in Fig.~\ref{fig:reductionoverview}(e)(f), to obtain the reduction from \tdm and prove Theorem~\ref{thm:main}.

\begin{proof}[Proof of Theorem~\ref{thm:main}]
For $x\in A\cup B\cup C$, let $k_x$ denote the number of hyperedges $t\in T$ that contain $x$.
For each $a\in A$, we place an element gadget $g_a$ with $k_a$ nodes, all of $A$ type; similarly, we place corresponding gadgets for each $b\in B$ and $c\in C$.
For every hyperedge $t\in T$, we place a splitting gadget $g_t$.
We connect the splitting gadgets to the element gadgets via wire gadgets. Wire gadgets can be thought of as transmitting information (i.e., the truth value) between two locations which may be far apart. For a hyperedge $t=(a,b,c)\in T$, we connect the $A$-type node of $g_t$ to one of the nodes of $g_a$, the $B$-type node to one of the nodes of $g_b$, and the $C$-type node to one of the nodes of $g_c$. As there are $k_a$ nodes of $g_a$ (and similarly for $g_b$ and $g_c$), they are each connected via wire gadgets to different splitting gadgets.
In a 2D geometry, such as for \cc or \surfc, different wires may need to cross each other. We cannot directly superpose the two wire gadgets since the overlap may result in minimal covers of the combined gadget that swaps the truth values between the two ends of a wire. Instead, we use a crossing gadget. Crossing gadgets are not needed for \tcnot since the wires can braid around each other in a 3D geometry.\footnote{Although not essential to the proof, we note that the syndrome created in this way is always satisfiable by \emph{some} error. This is because the defects in the wire, crossing, and splitting gadgets can be locally removed by the minimal cover with all nodes set to FALSE. In each of the three decoding scenarios, an element gadget contains an odd number of defects of one type. Since there are $r$ element gadgets of each type, using errors that create one defect of each type, one can find some error that removes all those defects.}

Suppose there is a perfect matching $M$ in the \tdm instance. For each hyperedge $t\in M$ chosen in the matching, we take the minimal cover $E_{g_t}$ that assigns TRUE to all nodes of $g_t$. For all other hyperedges $t\in T\setminus M$, we use $E_{g_t}$ that assigns FALSE to all nodes of $g_t$. A consistent choice of minimal covers of the wire gadgets and crossing gadgets connecting to these splitting gadgets exists. Because $M$ is a perfect matching, each element $x\in A\cup B\cup C$ will be in exactly one $t\in M$. We use the minimal cover $E_{g_x}$ of $g_x$ that sets the node connected (via a wire gadget) to the splitting gadget $g_t$ to TRUE while setting rest to FALSE. This is consistent because $x$ is only in one hyperedge; all the nodes set to FALSE are necessarily connected via wire gadgets to splitting gadgets corresponding to hyperedges in $T\setminus M$. Because $E=\bigcup_g E_g$ is consistent at every node, its total weight is $|E| = |\sigma| + m$.

Conversely, suppose an error $E$ with weight $|E|\le |\sigma|+m$ exists. By Lemma~\ref{lem:errorweightlowerbound}, it has weight $|E|=|\sigma|+m$ and the connected covers $E_g$ are minimal covers of each of the gadgets. Let $G_M$ be the set of splitting gadgets whose nodes are set to TRUE and $M\subseteq T$ be the corresponding subset of hyperedges.
Each element gadget has exactly one node set to TRUE, which must be connected via a wire gadget to some splitting gadget in $G_M$. This shows that each element of $A\cup B\cup C$ is in exactly one of the hyperedges in $M$, i.e., $M$ is a perfect matching for the \tdm instance.
\end{proof}

\section{Gadget construction and placement}
\label{sec:gadgets}
We now describe the details of how the gadgets are implemented in the three decoding settings and show that they have the required properties in Assumption~\ref{ass:gadgetproperties}.
We also explain how to place them so that they are sufficiently separated while having the required connections, satisfying Assumption~\ref{ass:separation}.
For all the gadgets, Assumption~\ref{ass:gadgetdefectsubset} can be verified by inspection, so we will not elaborate on it.

\subsection{Color code gadgets}
\label{sec:CC}
We start with the color code on a hexagonal lattice, whose decoding problem for $Z$ errors was very recently shown to be NP-hard in Ref.~\cite{walters2026CCNPhard}.
That work also introduced the splitting gadget (called RGB-duplicator subgadget), the wire gadget, and the crossing gadget between wires of the same or different colors (called full crossing gadget and multi-color crossing subgadget, respectively). These gadgets, which we illustrate in Figs.~\ref{fig:gadgetsCC} and~\ref{fig:crossinggadgetsametype} for completeness, were shown to have the required properties in Assumption~\ref{ass:gadgetproperties}. There is no error excess for any of these gadgets.

\begin{figure}[htpb]
    \centering
    (a)\includegraphics[width=0.22\textwidth, page=1, trim={2.5cm 6cm 12cm 8cm}, clip = True]{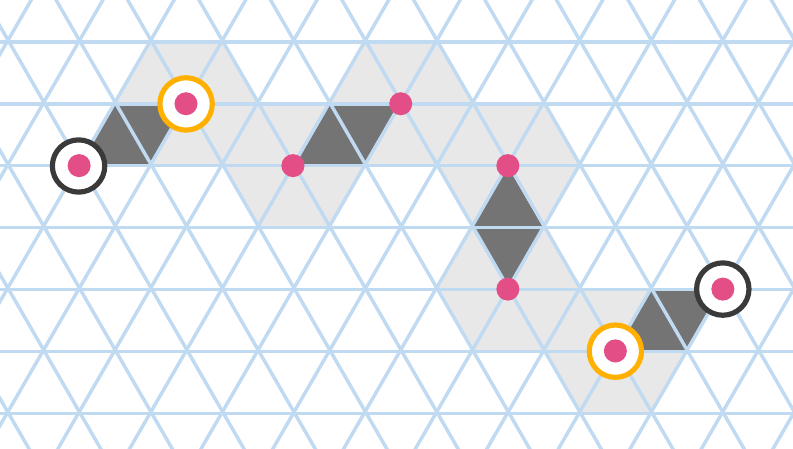}
    (b)\includegraphics[width=0.22\textwidth, page=2, trim={2.5cm 6cm 12cm 8cm}, clip = True]{Figures/CCgadgets.pdf}
    (c)\includegraphics[width=0.22\textwidth, page=3, trim={6cm 5.5cm 9.5cm 5.5cm}, clip = True]{Figures/CCgadgets.pdf}
    (d)\includegraphics[width=0.22\textwidth, page=4, trim={6cm 5.5cm 9.5cm 5.5cm}, clip = True]{Figures/CCgadgets.pdf}
    (e)\includegraphics[width=0.22\textwidth, page=5, trim={5cm 3cm 9.5cm 5.5cm}, clip = True]{Figures/CCgadgets.pdf}
    (f)\includegraphics[width=0.22\textwidth, page=6, trim={5cm 3cm 9.5cm 5.5cm}, clip = True]{Figures/CCgadgets.pdf}
    (g)\includegraphics[width=0.22\textwidth, page=7, trim={5cm 3cm 9.5cm 5.5cm}, clip = True]{Figures/CCgadgets.pdf}
    (h)\includegraphics[width=0.22\textwidth, page=8, trim={5cm 3cm 9.5cm 5.5cm}, clip = True]{Figures/CCgadgets.pdf}
    \caption{(a)(b) The wire gadget, (c)(d) splitting gadget, and (e)-(h) crossing gadget between wires of different defect types in the \cc decoding problem, as introduced in Ref.~\cite{walters2026CCNPhard}.
    The different minimal covers are shown.
    In all figures, dots represent defects (unsatisfied $X$ stabilizers in the color code), with the nodes circled in yellow and the partner nodes (which are defects in another gadget) circled in black. The neighborhood of the gadget is highlighted in light gray.
    }
    \label{fig:gadgetsCC}
\end{figure}

\begin{figure}[htpb]
    \centering
    \includegraphics[width=0.6\textwidth, trim={0cm 5.5cm 0cm 5.5cm}, clip = True]{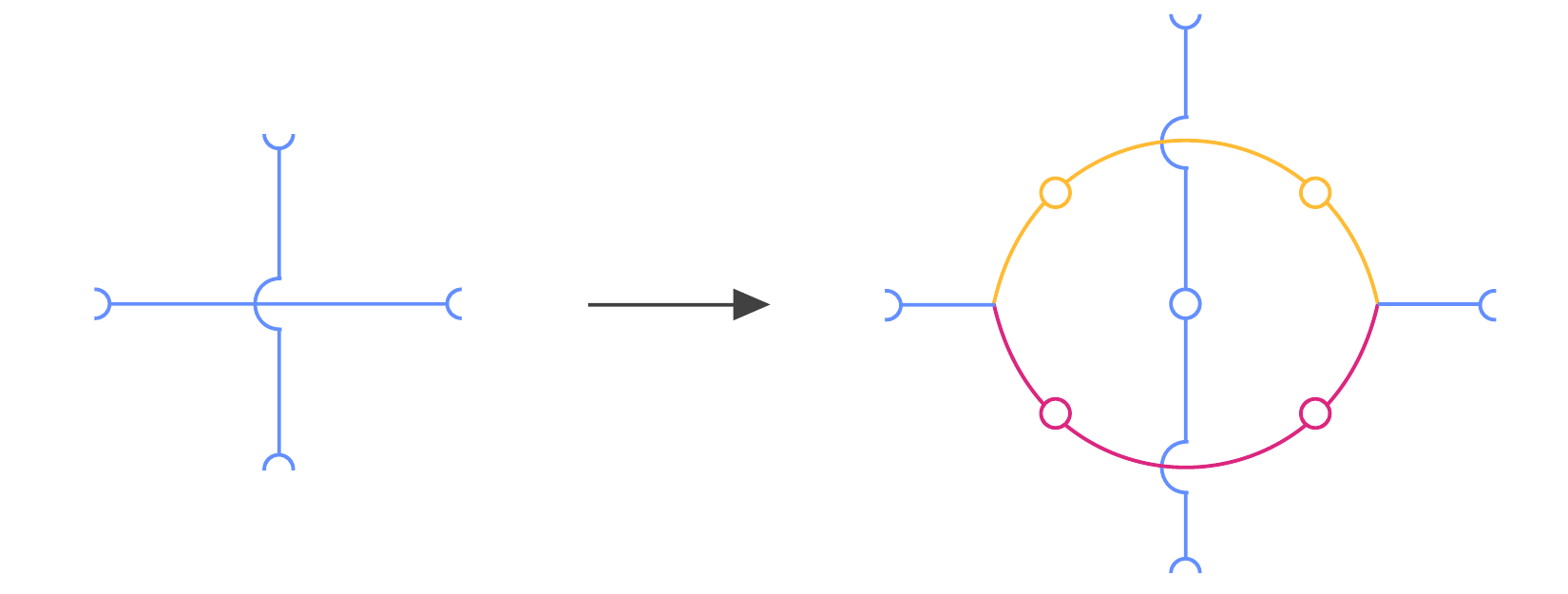}
    \caption{The crossing gadget for wires of the same defect type is constructed from two splitting gadgets and two crossing gadgets between wires of different defect types.
    }
    \label{fig:crossinggadgetsametype}
\end{figure}

It remains to construct the element gadget, which we depict in Fig.~\ref{fig:setelementCC}. The error excess is bounded using the following argument.

\begin{lemma}
    \label{lem:isolateddefects}
    If a gadget $g$ consists of defects of distance at least four from each other in the Tanner graph (which we call \emph{isolated defects}), then it has error excess $m_g\ge 0$.
\end{lemma}
\begin{proof}
    There must be an error adjacent to every defect in $g$, which is necessarily in $R_g$ because it contains the 1-neighborhood of $g$. The statement then follows because the condition on the distances between the defects means that each error can be adjacent to at most one defect.
\end{proof}
Error configurations $E_g$ similar to the one in Fig.~\ref{fig:setelementCC} achieve $|E_g\cap R_g|=|\sigma|$, showing that $m_g=0$ and that exactly one node is TRUE in a minimal cover. Note that $E_g\subseteq R$ because every error is adjacent to some defect.

\begin{figure}[htpb]
    \centering
    \includegraphics[width=0.45\textwidth,trim={2cm 3.2cm 7.5cm 7.5cm}, clip = True]{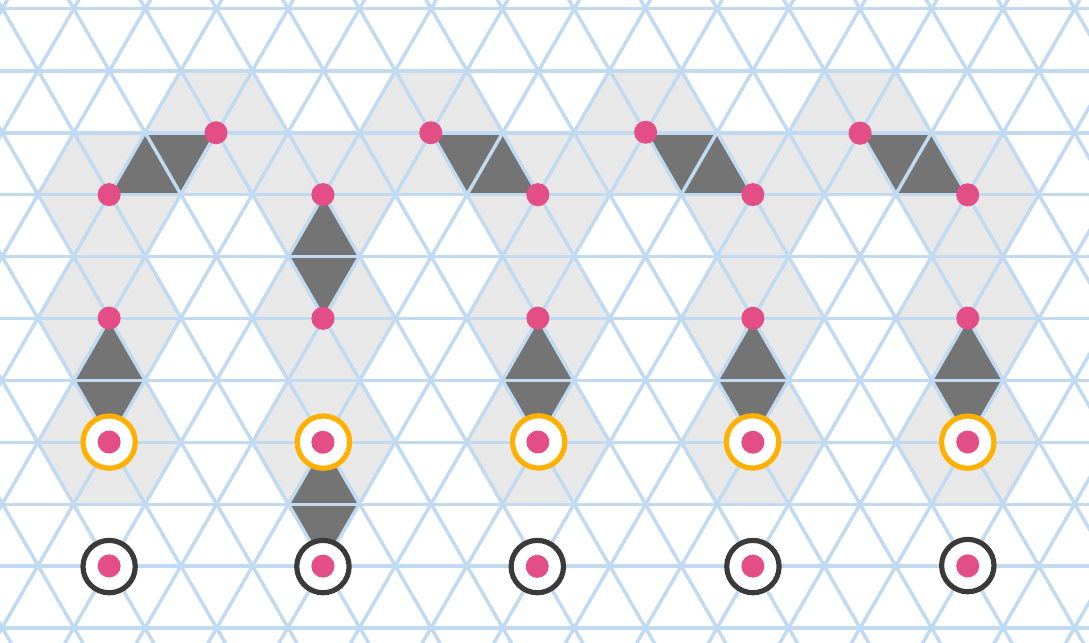}
    \caption{The element gadget with $R$ defects. Here, we illustrate the minimal cover with the second node set to TRUE.}
    \label{fig:setelementCC}
\end{figure}

We place the gadgets as shown in Fig.~\ref{fig:reductionoverview}(e). This allows the gadgets to be well separated from each other, satisfying Assumption~\ref{ass:separation}. In particular, note that each gadget takes up a constant amount of space, except possibly the wire gadgets. The width of total syndrome configuration is $O(s)$. We implement layers of swaps to route the $3s$ wires from the element gadgets to the appropriate splitting gadgets. With 1D connectivity, any permutation of $3s$ elements can be implemented with a depth-$O(s)$ sorting network consisting of layers of nearest-neighbor swaps~\cite{knuth98sorting}.
Therefore, the height of the syndrome configuration is also $O(s)$. This shows that the decoding problem has polynomial size, as the syndrome can be embedded in a color code lattice with $O(s^2)$ qubits.

\subsection{Surface code gadgets}
\label{sec:SC}
Next, we construct the gadgets for the surface code on a square lattice. 
The wire gadgets $g$ for $Z$ defects are shown in Fig.~\ref{fig:WireZgadget}, with similar ones for $X$ defects. Because the defects are isolated, Lemma~\ref{lem:isolateddefects} implies $m_g\ge 0$ for the gadget in Fig.~\ref{fig:WireZgadget}(a)(b). The configurations shown achieve $|E_g\cap R_g|=|g|$, so those are minimal covers and $m_g=0$.
The gadget in Fig.~\ref{fig:WireZgadget}(c)(d) can be seen to have error excess $m_g=-1$ based on the minimal covers shown in that figure and Lemma~\ref{lem:isolateddefects} (removing one of the $X$ defects gives a set of isolated defects).
We will often use a parity argument which says that any error in the bulk of the surface code causes an even number of $X$ defects and an even number of $Z$ defects (assuming that we decompose any $Y$ defect into an $X$ and a $Z$ defect)---to reason about the truth values in a minimal cover.
In this case, it implies that the nodes must either both be TRUE or both be FALSE, as desired in Assumption~\ref{ass:gadgetproperties}.
Using a combination of the two gadgets allows us to place the nodes of a $Z$ wire gadget at any vertex on the lattice.

\begin{figure}[htpb]
    \centering
    (a)\includegraphics[width=0.4\textwidth, page=1, trim={1cm 8.5cm 11.5cm 6.2cm}, clip = True]{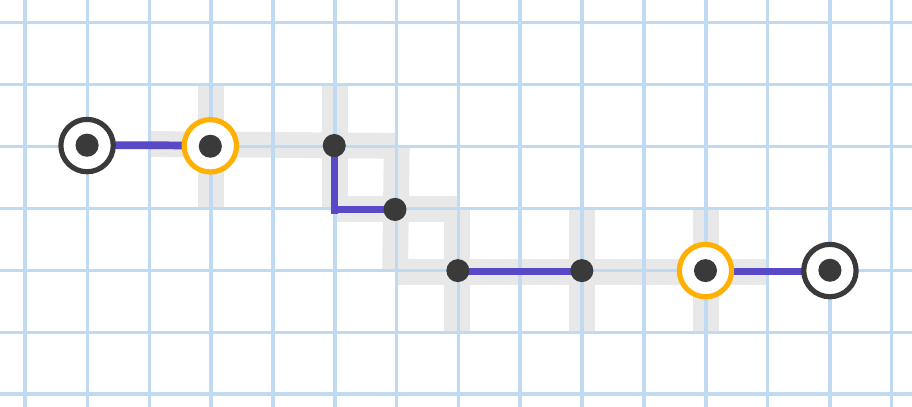}
    (b)\includegraphics[width=0.4\textwidth, page=2, trim={1cm 8.5cm 11.5cm 6.2cm}, clip = True]{Figures/WireZgadget.pdf}
    (c)\includegraphics[width=0.4\textwidth, page=1, trim={1cm 4.2cm 11.5cm 12.5cm}, clip = True]{Figures/WireZgadget.pdf}
    (d)\includegraphics[width=0.4\textwidth, page=2, trim={1cm 4.2cm 11.5cm 12.5cm}, clip = True]{Figures/WireZgadget.pdf}
    \caption{Two types of wire gadgets with $Z$ defects in the surface code. Minimal covers with both nodes TRUE and both nodes FALSE are shown in (a)(c) and (b)(d), respectively.
    In our surface code figures, blue strings denote $Z$ errors and red strings denote $X$ errors (with $Y$ errors on any overlap).
    Having access to both gadgets in (a)(b) and (c)(d) allows us to place the nodes at any vertex on the lattice.
    }
    \label{fig:WireZgadget}
\end{figure}

\begin{figure}[htpb]
    \centering
    (a)\includegraphics[width=0.4\textwidth, page=1, trim={1cm 8cm 9.5cm 2cm}, clip = True]{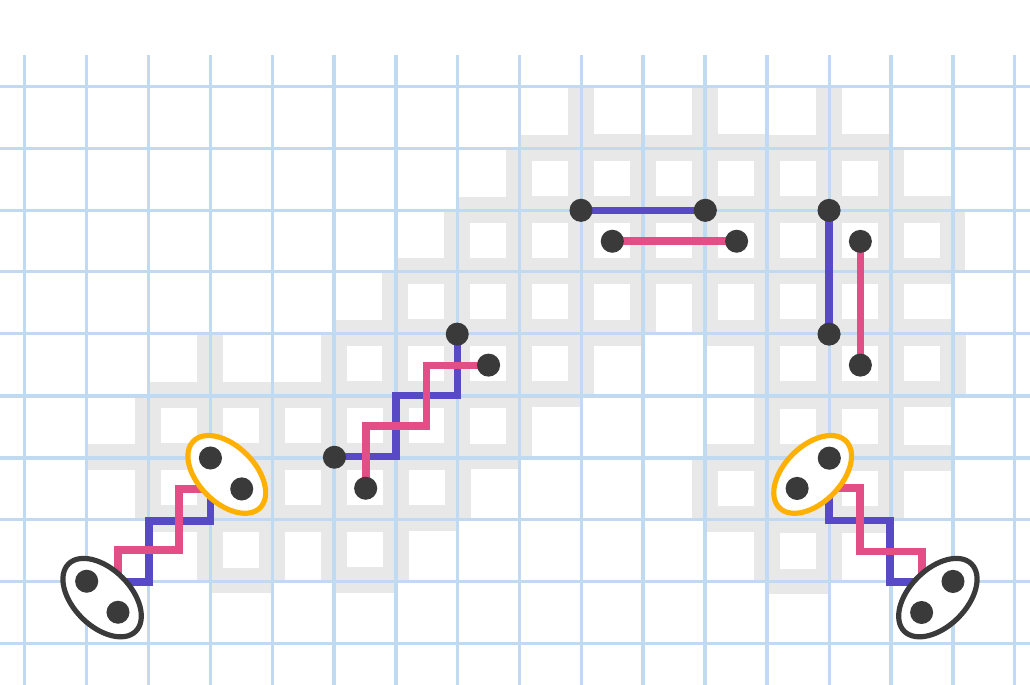}
    (b)\includegraphics[width=0.4\textwidth, page=2, trim={1cm 8cm 9.5cm 2cm}, clip = True]{Figures/WireYgadget.pdf}
    \caption{The $Y$ wire gadget in the surface code with minimal covers where (a) both nodes are TRUE and (b) both nodes are FALSE.}
    \label{fig:YWire}
\end{figure}

Fig.~\ref{fig:YWire} illustrates the $Y$ wire gadget. Although the defects are not isolated, we still have $m_g=0$ with the minimal covers shown. Indeed, the $Z$ defects in the gadget must be matched to each other, and similarly for the $X$ defects. For the diagonal part of the wire, minimum-weight matchings for both $X$ and $Z$ defects can be chosen that completely overlap, giving the combined minimum-weight configuration.
In the case of a horizontal (or vertical) segment, the configuration shown is minimal because the two $Z$ defects must be matched with a $Z$ string of length at least two and the two $X$ defects with an $X$ string of length at least two. A $Y$ error can reduce the required weight of one of these strings but not both.
As shown, there appear to be four nodes. By a parity argument, the two $X$ nodes must have the same truth value and the two $Z$ nodes must have the same truth value. In fact, all four values must be the same if there is at least one diagonal segment because the $X$ and $Z$ strings must overlap on a diagonal segment in a minimal cover, forcing the adjacent $X$ and $Z$ nodes to have the same truth value. Therefore, we can view the adjacent $X$ and $Z$ nodes as a combined $Y$ defect so that the $Y$ wire has two $Y$ nodes of the same truth value.
Note that the partner of a $Y$ node should be placed diagonally away from it, as is the case in Fig.~\ref{fig:YWire}, because its 2-neighborhood should cover half the error connecting to its partner when it is TRUE.

The splitting gadget is shown in Fig.~\ref{fig:splitWire}. The truth values of all three nodes are the same from a parity argument and the fact that a minimal cover must match the $Y$ defects using a $Y$ string. This gadget has an error excess of zero.

\begin{figure}[htpb]
    \centering
    (a)\includegraphics[width=0.3\textwidth, page=1, trim={0cm 4cm 14.5cm 4cm}, clip = True]{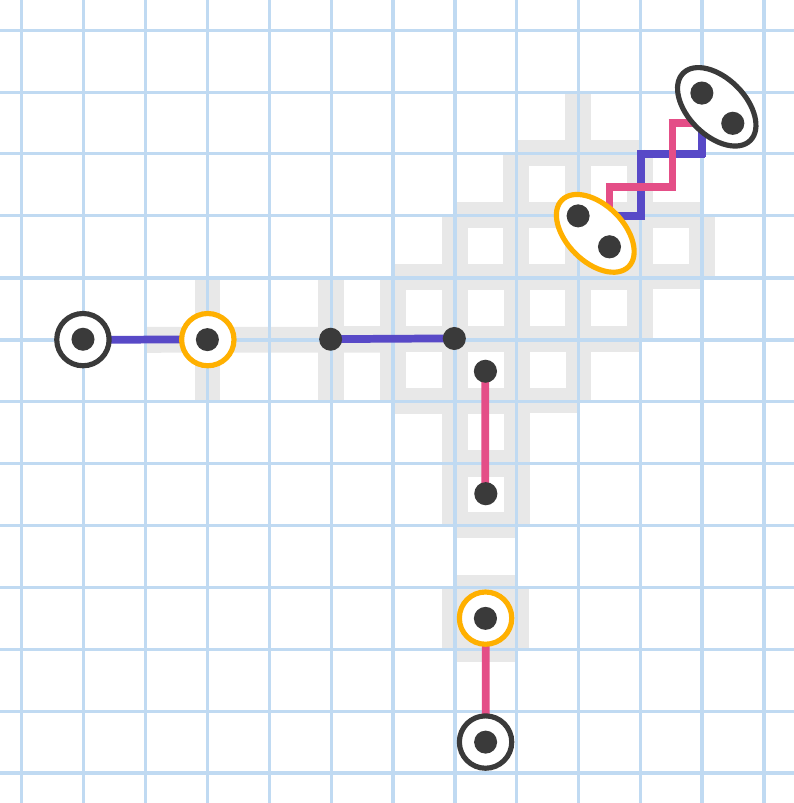}
    (b)\includegraphics[width=0.3\textwidth, page=2, trim={0cm 4cm 14.5cm 4cm}, clip = True]{Figures/SCsplittinggadget.pdf}
    \caption{The splitting gadget in the surface code. Minimal covers with (a) all three nodes TRUE and (b) all three nodes FALSE are shown.
    }
    \label{fig:splitWire}
\end{figure}

Fig.~\ref{fig:XZcrossinggadget} shows a crossing gadget between an $X$ wire and a $Z$ wire. A parity argument shows that the two $X$ defects have the same truth values and the two $Z$ defects have the same truth value. The $Z$ defects are isolated if we do not consider the $X$ defects, and vice versa. When combined, some $X$ and $Z$ defects in the center are adjacent, which allows an $X$ string and a $Z$ string to overlap by one in a minimal cover in each of the four combinations of truth values. This shows that the proposed set of defects is a valid crossing gadget with error excess $-1$.

\begin{figure}[htpb]
    \centering
    (a)\includegraphics[width=0.3\textwidth, page=1, trim={0cm 5cm 12.5cm 5cm}, clip = True]{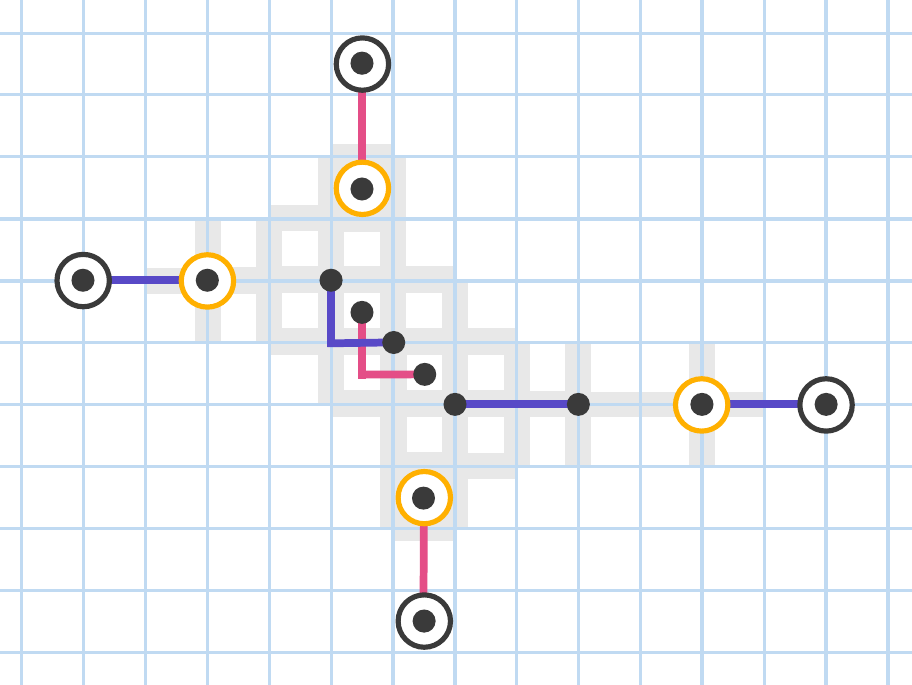}
    (b)\includegraphics[width=0.3\textwidth, page=2, trim={0cm 5cm 12.5cm 5cm}, clip = True]{Figures/XZcrossinggadget.pdf}\\
    (c)\includegraphics[width=0.3\textwidth, page=3, trim={0cm 5cm 12.5cm 5cm}, clip = True]{Figures/XZcrossinggadget.pdf}
    (d)\includegraphics[width=0.3\textwidth, page=4, trim={0cm 5cm 12.5cm 5cm}, clip = True]{Figures/XZcrossinggadget.pdf}
    \caption{The crossing gadget between an $X$ and a $Z$ wire in the surface code. Minimal covers with the different valid combinations of truth values of the nodes are shown. The center of the gadget allows an error to cover it which has weight one less than the number of defects covered, regardless of the truth values of the nodes.
    }
    \label{fig:XZcrossinggadget}
\end{figure}

A crossing gadget between a $Y$ wire and a $Z$ wire is shown in Fig.~\ref{fig:YZcrossinggadget}. The error configurations shown are the minimal covers with error excess $1$, which give the same truth values on the $Y$ nodes and on the $X$ nodes. Minimality of the covers can be seen from the fact that the outside portion of the gadget consists of $Y$ and $Z$ wires and analyzing the possible error configurations in the center.\footnote{For example, in the center of Fig.~\ref{fig:YZcrossinggadget}(a) consisting of four $Z$ defects and two $X$ defects, any matching for the $Z$ defects involves at least six errors. Since the two $X$ defects cannot be matched with a string that completely overlaps with any of the weight-six configurations, the minimum weight is seven, with one such configuration shown.}
A crossing between an $X$ wire and a $Y$ wire is similar.

\begin{figure}[htpb]
    \centering
    (a)\includegraphics[width=0.3\textwidth, page=1, trim={1cm 4cm 13.3cm 4cm}, clip = True]{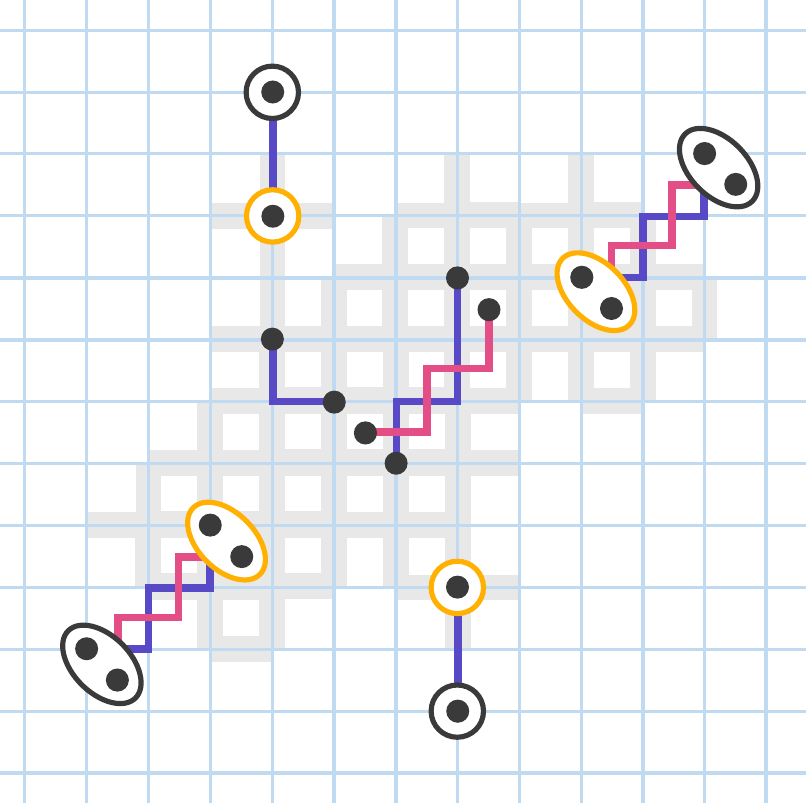}
    (b)\includegraphics[width=0.3\textwidth, page=2, trim={1cm 4cm 13.3cm 4cm}, clip = True]{Figures/YZcrossinggadget.pdf}\\
    (c)\includegraphics[width=0.3\textwidth, page=3, trim={1cm 4cm 13.3cm 4cm}, clip = True]{Figures/YZcrossinggadget.pdf}
    (d)\includegraphics[width=0.3\textwidth, page=4, trim={1cm 4cm 13.3cm 4cm}, clip = True]{Figures/YZcrossinggadget.pdf}
    \caption{The crossing gadget between a $Y$ and a $Z$ wire in the surface code. In each of the four combinations of truth values of the nodes (the minimal covers shown), the error weight is one more than the number of defects covered.
    }
    \label{fig:YZcrossinggadget}
\end{figure}

To construct a crossing between two wires $w_1$ and $w_2$ of the same defect type, we can first split $w_1$ into two wires $w_1'$, $w_1''$ of the other two defect types (and the same truth value), cross those with $w_2$, and then recombine them into one wire of the original type. This has been already shown in Fig.~\ref{fig:crossinggadgetsametype}.

Finally, the element gadget with $Z$ nodes for the $C$ set is shown in Fig.~\ref{fig:setelementSC}(a). By an analysis similar to the one for the corresponding gadget in the color code, we see that exactly one node is TRUE in a minimal cover and there is no error excess. The element gadget with $X$ nodes is analogous. Fig.~\ref{fig:setelementSC}(b) shows the gadget with $Y$ nodes for the $B$ set, which is essentially a juxtaposition of the element gadgets with $X$ nodes and the one with $Z$ nodes. Note that the diagonal connection between the node and its partner ensures that the nodes at the same position must be chosen to be TRUE when considering the gadget as two separate element gadgets.

\begin{figure}[htpb]
    \centering
    (a)\includegraphics[width=0.43\textwidth, page=1, trim={0.7cm 5cm 7cm 7cm}, clip = True]{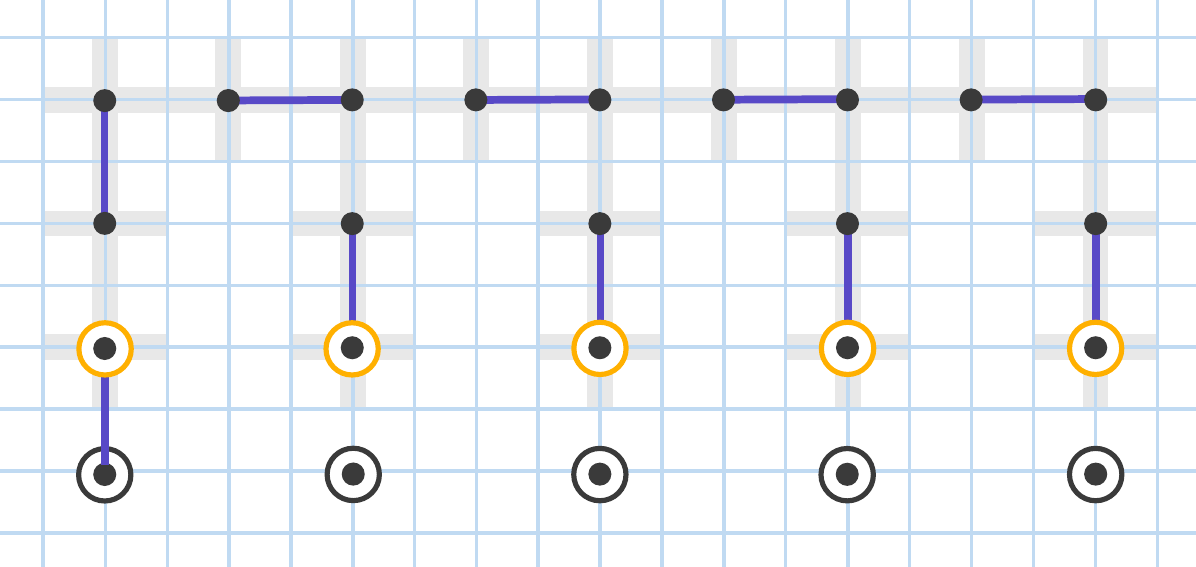}
    (b)\includegraphics[width=0.51\textwidth, page=2, trim={0.7cm 4.2cm 3.5cm 6cm}, clip = True]{Figures/setelementgadgetSC.pdf}
    \caption{The element gadget in the surface code with (a) $Z$ nodes and (b) $Y$ nodes in the surface code. Exactly one of the nodes will be TRUE in a minimal error.}
    \label{fig:setelementSC}
\end{figure}

By placing the gadgets sufficiently far from each other in Fig.~\ref{fig:reductionoverview}(e) and using an $O(s)$-depth sorting network, we satisfy Assumption~\ref{ass:separation} and obtain a syndrome embedded in a surface code with $O(s^2)$ qubits.

\subsection{Transversal CNOT decoding gadgets}
\label{sec:tCNOT}

We now consider decoding of two square-lattice surface code patches coupled via a transversal logical CNOT gate in an intermediate time step. The problem is to find the minimum combined weight of $Z$ qubit errors and measurement bit-flip errors, given a set of triggered $X$ detectors. Away from the time slice $t_G$ of the transversal CNOT gate, the decoding graph is a 3D cubic lattice, so we may use essentially the same gadgets as the ones in Sec.~\ref{sec:SC} with $Z$ defects. The cubic lattice is bipartite, and we will only place nodes on one of the vertex subsets.\footnote{We specify this because the nodes of the wire gadget in Fig.~\ref{fig:WireZgadget}(a)(b) are always in the same vertex subset.} Recall that the defect type in this problem is based on the spacetime region where it is located.

The wire gadget (for any of the three defect types) is analogous to the one in Fig.~\ref{fig:WireZgadget}(a)(b) but placed in three dimensions.
Crossing gadgets are not needed in this setting because wire gadgets can be braided around each other in a 3D geometry.
The element gadget is identical to the one in Fig.~\ref{fig:setelementSC}(a).

The main difference in this decoding scenario is the implementation of the splitting gadget. The splitting gadget may only be placed across the time slice $t_G$, shown in Fig.~\ref{fig:duplicatortCNOT}. This gadget consists of isolated defects, so by Lemma~\ref{lem:isolateddefects}, the error configurations shown are minimal covers and the error excess is zero. Using a parity argument and the fact that the hyperedges in the decoding graph create one defect of each type shows that the nodes in the gadget must all be TRUE or all be FALSE.

\begin{figure}[htpb]
    \centering
    (a)\includegraphics[width=0.45\textwidth, page=1,trim = {1cm 3cm 5cm 2.5cm},clip=true]{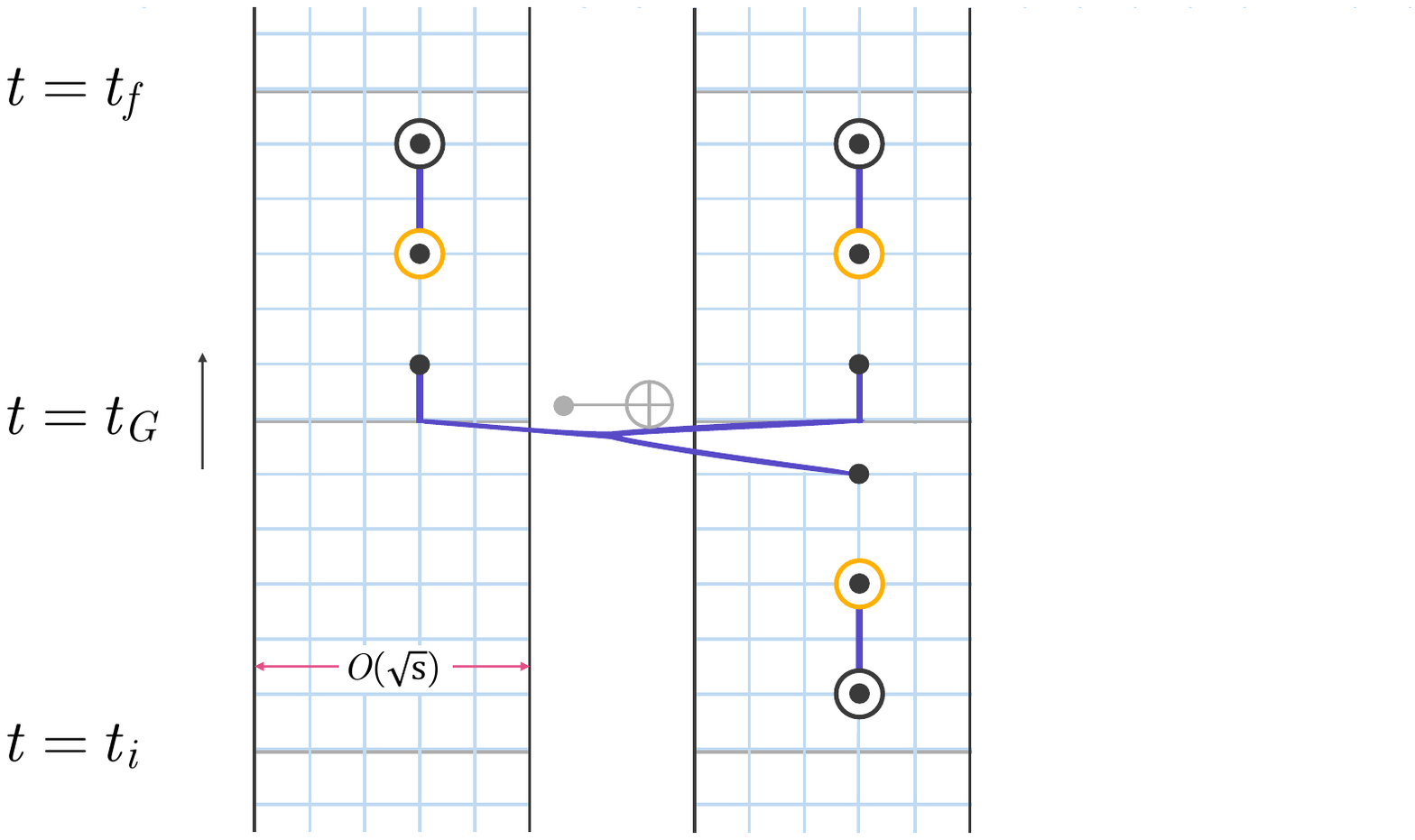}
    (b)\includegraphics[width=0.45\textwidth, page=2,trim = {1cm 3cm 5cm 2.5cm},clip=true]{Figures/spacetimegadget2.pdf}
    \caption{The splitting gadget in the \tcnot decoding problem. A (1+1)D projection of the (2+1)D surface code is illustrated, with time going upward. The hyperedge going across the two surface codes arises from a single measurement error (other possible measurement errors not shown), and all three wires are at the same spatial locations of their respective surface codes.
    Minimal covers are shown with (a) all nodes TRUE and (b) all nodes FALSE.
    }
    \label{fig:duplicatortCNOT}
\end{figure}

The placement of the gadgets in the \tcnot problem, shown in Fig.~\ref{fig:reductionoverview}(f), is slightly different than in the case of 2D decoding graphs for \cc and \surfc. We place all the element gadgets for the elements in $A$ at the time slice $t_i$ in the spacetime region hosting $A$ defects, and similarly for the elements in $B$ and $C$ but at time slice $t_f$.\footnote{If perfect measurements at times $t_i-\frac{1}{2}$ and $t_f+\frac{1}{2}$ are not assumed, we could instead place the element gadgets further in the spacetime bulk, e.g., at time slices $(t_i+t_G)/2$ and $(t_G+t_f)/2$.} All the splitting gadgets for the hyperedges $t\in T$ are on the time slice of the transversal CNOT gate. We then use wire gadgets to connect the element gadgets to the appropriate splitting gadgets. In other words, the three wires to be merged are routed to same spatial locations of each region near the transversal CNOT time slice. For the gadgets placed according to a size-$s$ \tdm problem to be well separated and satisfy Assumption~\ref{ass:separation}, the element gadgets and splitting gadgets should occupy a spatial region of size $O(\sqrt s)\times O(\sqrt s)$. The wires need to be routed while maintaining separation. One way to do this is to use a 2D sorting network, which can implement an arbitrary permutation in depth $O(\sqrt s)$ using nearest-neighbor swaps~\cite{TK77sorting}. The element gadgets are therefore placed at temporal distance $t_f-t_G=t_G-t_i=O(\sqrt s)$ from the time slice of the transversal CNOT gate.
Thus, the decoding problem has polynomial size, as the syndrome can be placed in a spacetime region with $O(s^{3/2})$ error locations.

\section{Hardness of decoding a logical operator}
In the reduction used to prove Theorem~\ref{thm:main}, the syndrome to be decoded is placed in the bulk of a large lattice. Even in the case that no perfect matching exists in the \tdm instance, the minimum-weight error would likely also be supported in the bulk of the lattice. Then, there would be no difference between the logical effect of the correction compared to the case where a perfect matching exists, and one might hope that it is efficient to decode a logical operator. However, we show that determining the logical effect of applying a minimum-weight correction is also NP-hard.

\begin{theorem}
    \label{thm:logical}
    In each of the three decoding scenarios, the problem of deciding if a minimum-weight error anticommutes with a given representative of a logical operator is NP-hard.
\end{theorem}

The idea of the proof is to use a modified syndrome pattern in the reduction from \tdm. For the modified syndrome, the minimum-weight error will be similar to the error for the original syndrome when a perfect matching exists. However, the modified syndrome allows us to control what a minimum-weight error looks like when no perfect matching exists. In particular, the minimum weight of a valid error exceeds the minimum weight when there is a perfect matching by one, and that error will have the opposite effect on a chosen logical operator.

\begin{figure}[htpb]
    \centering
    (a)\includegraphics[width=0.4\textwidth, page=1, trim = {4cm 2.5cm 10cm 3cm},clip=true]{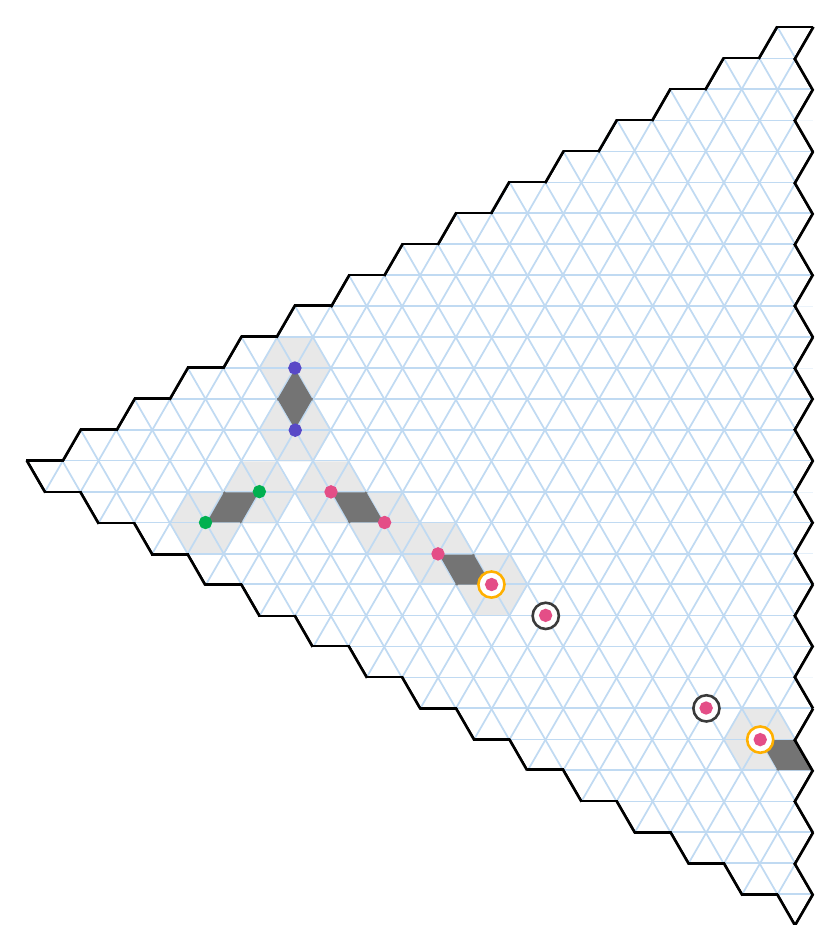}
    (b)\includegraphics[width=0.4\textwidth, page=2, trim = {4cm 2.5cm 10cm 3cm},clip=true]{Figures/CCspecialwire.pdf}
    (c)\includegraphics[width=0.4\textwidth, trim = {1.38cm 6.2cm 0.27cm 10.4cm},clip=true]{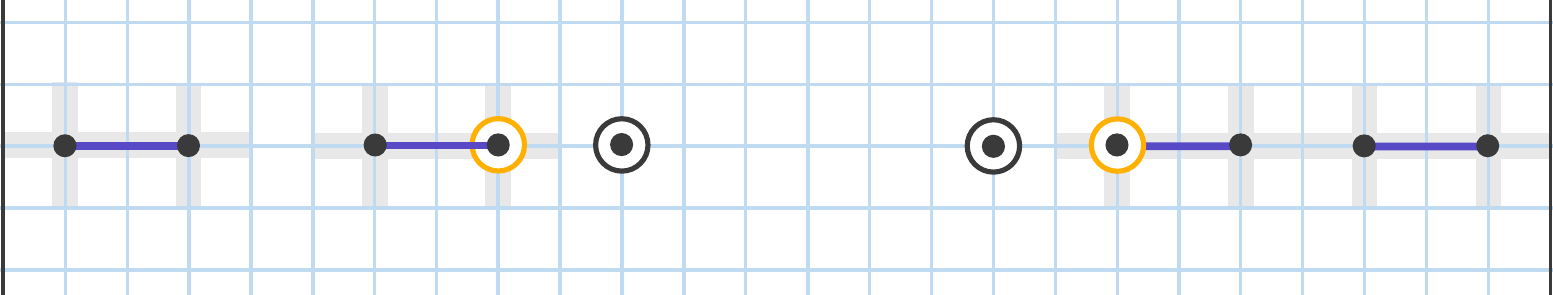}
    (d)\includegraphics[width=0.4\textwidth, trim = {1.38cm 10.4cm 0.27cm 6.2cm},clip=true]{Figures/logicalwireb2.pdf}
    (e)\includegraphics[width=0.4\textwidth, trim = {0cm 7cm 0cm 7cm},clip=true]{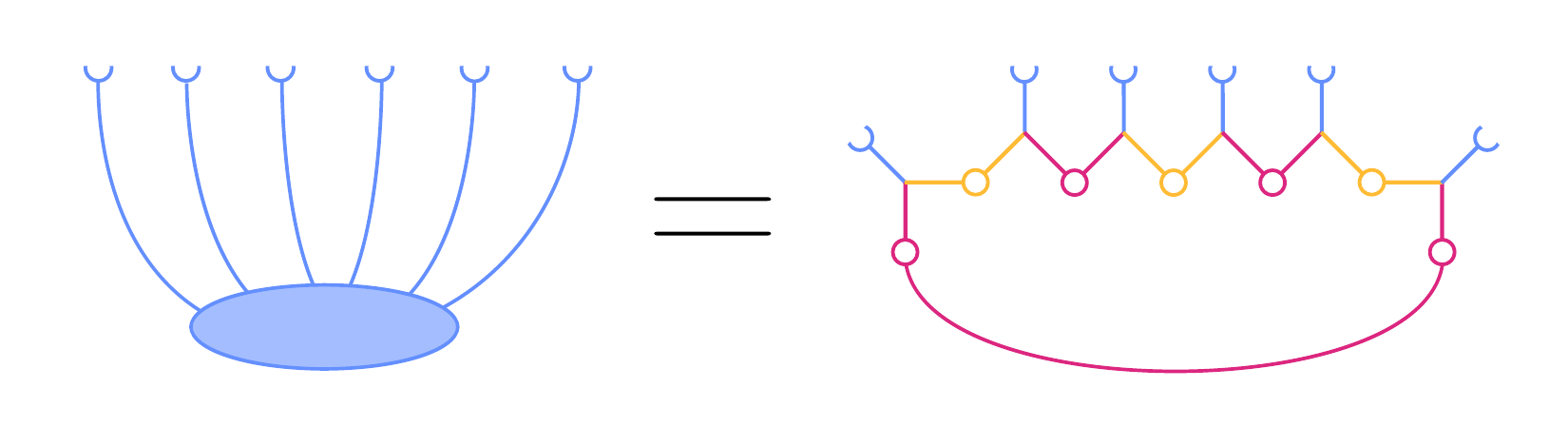}
    (f)\includegraphics[width=0.7\textwidth]{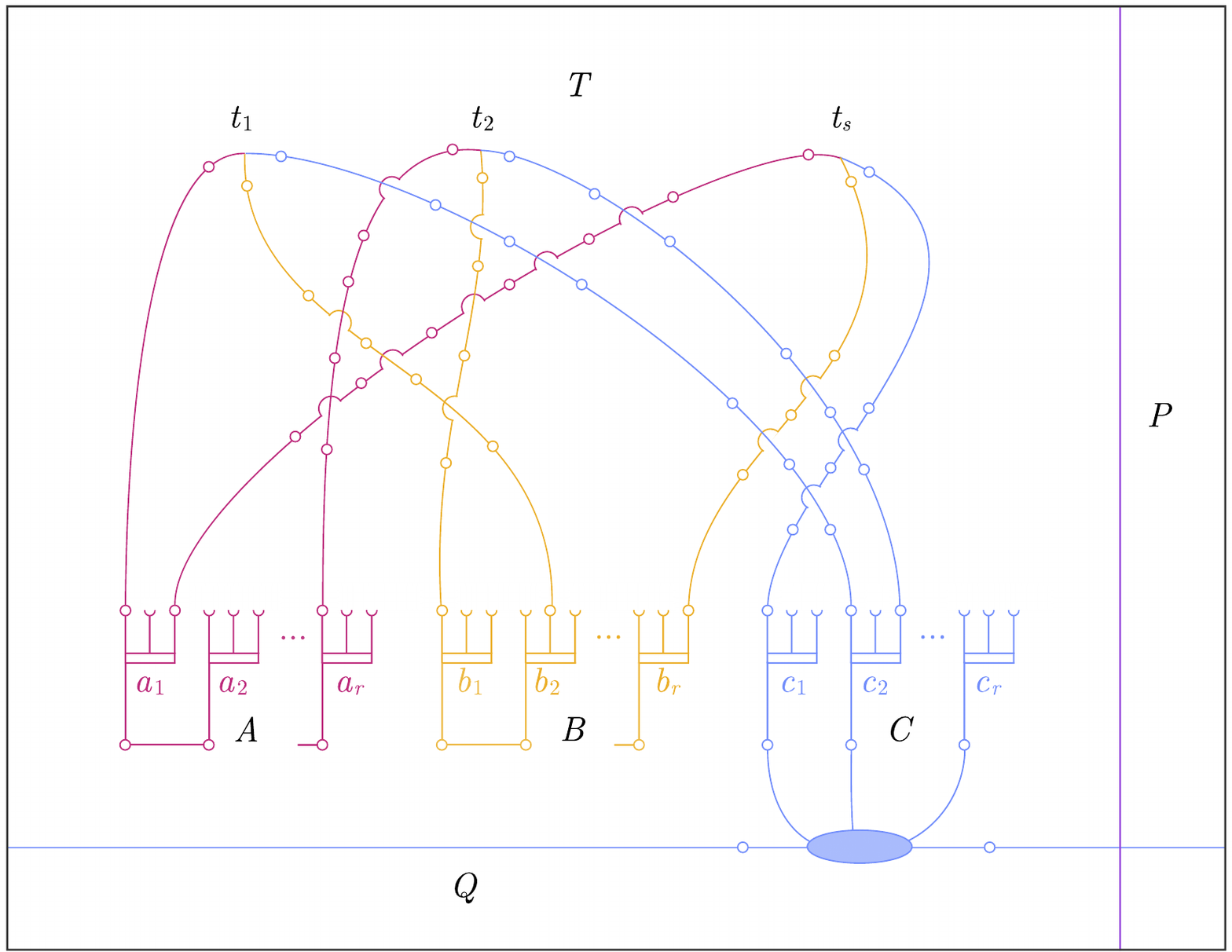}
    \caption{The special wire gadget in (a)(b) the color code and (c)(d) the surface code. When both nodes are set to FALSE, as in (a)(c), the error configuration contains one fewer error outside $R$ compared to when both nodes are set to TRUE, as in (b)(d).
    Not all vertices on the boundaries of the color code support stabilizers.
    (e) The multi-wire splitting gadget, implemented using splitting gadgets.
    (f) The syndrome to be decoded when reducing \tdm to the problem of determining if the minimum-weight error flips a given logical operator (purple vertical line). 
    The long, blue horizontal line is the special wire gadget.
    }
    \label{fig:specialwiregadget}
\end{figure}

\begin{proof}[Proof of Theorem~\ref{thm:logical}]
    We consider decoding a QEC code with an odd-weight string operator $Q$. In the color code, this can be achieved using an odd-distance triangular color code with open boundary conditions. In the surface code (with or without measurement errors), we can use an odd distance code with either open or periodic boundary conditions. We will decode the logical operator $P$ (the purple line in Fig~\ref{fig:specialwiregadget}(f)), which intersects $Q$ and is sufficiently far from the original gadgets.

    From the operator $Q$, we create a special wire gadget $g$.
    For this gadget, there are two error configurations $E_g, E_g'$ whose syndrome is contained in $g$ and its partner nodes that minimize $|E\cap R_g|$. These configurations have both nodes set to FALSE or both nodes set to TRUE, respectively.
    However, $E_g'$ contains one more error outside of $R$ (the neighborhoods of all gadgets) than $E_g$, which would result in a higher-weight global error if used.
    Such an operator for the color code was presented in Ref.~\cite{walters2026CCNPhard} and is shown in Fig.~\ref{fig:specialwiregadget}(a)(b).
    The special wire gadget for the surface code (in the \tcnot decoding scenario, inserted at a constant time slice) is shown in Fig.~\ref{fig:specialwiregadget}(c)(d).
    The two error configurations have opposite effects on the logical operator $P$.

    We also introduce a multi-wire splitting gadget (called the variable gadget in Ref.~\cite{walters2026CCNPhard}). This is a gadget with $r+2$ nodes, all of the same type, that have the same truth value in any minimal cover. Note that $r$ must be even. It can be implemented with the usual splitting gadgets, as shown in Fig.~\ref{fig:specialwiregadget}(e). For the \tcnot decoding scenario, this gadget may only be placed near the time slice $t_G$ because of the same restriction for the splitting gadgets.

    The new syndrome configuration is shown in Fig.~\ref{fig:specialwiregadget}(f).
    Starting from the original placements of the gadgets, we add one additional node to each element gadget. Within the set $A$, we pair up these new nodes, connecting them via wire gadgets, and repeat for the set $B$. Note that we can assume without loss of generality that $r=|A|=|B|=|C|$ is even in the \tdm instance; if not, adding extra elements $a_0\in A$, $b_0\in B$, and $c_0\in C$ and a single extra hyperedge $(a_0, b_0, c_0)\in T$ would give an equivalent \tdm instance. The extra nodes for the set $C$ are connected via wire gadgets to the multi-wire splitting gadget. This requires routing the wires from time slice $t_f$ to time slice $t_G$ in the \tcnot decoding scenario. The other two nodes of the multi-wire splitting gadget are connected to the special wire gadget, shown as the blue horizontal line.

    Consider the case where the \tdm instance has a perfect matching. Then we may use the same minimal cover as in the decision problem. All the extra nodes of the element gadgets are set to FALSE, as well as the nodes of the multi-wire splitting gadget and the special wire gadget. This gives a globally minimal error configuration with $|E|=|\sigma|+m$. Here, $m=\sum_g m_g+1$ for the color code and $m=\sum_g m_g$ for the surface code, which result from the special wire gadget and a modification of Eq.~\eqref{eq:errorwtbound}.

    On the other hand, if an error $E$ with weight $|E|=|\sigma|+m$ existed, the nodes of the special wire gadget, and hence those of the multi-wire splitting gadget and the extra nodes of the element gadgets for the set $C$, must be FALSE. The TRUE nodes of those element gadgets would connect to $r$ splitting gadgets, which would then connect to each of the elements in $A$ and $B$ exactly once, giving a perfect matching of the \tdm instance.
    
    However, in the case that the \tdm instance does not have a perfect matching, we can find an error $E$ of weight $|E|=|\sigma|+m+1$.
    To do this, we set all of the original nodes of the element gadget to FALSE and the new nodes to TRUE. This allows a minimal cover of all the splitting gadgets associated with the hyperedges in $T$, as they are all set to FALSE. The nodes of the multi-wire splitting gadget and the special wire gadget are then set to TRUE, which gives an error configuration with one extra weight compared to if they had been set to FALSE. This shows that whether the minimum-weight error flips the logical operator $P$ corresponds to whether the \tdm instance has a perfect matching.
\end{proof}

\section{Discussion}

In this work, we proved that three quintessential minimum-weight decoding problems, \cc, \surfc, and \tcnot, are NP-hard via a reduction from \tdm.
Our results show that computational intractability already arises in basic and practically relevant decoding problems in QEC, including settings that are central to both quantum memories and logical circuit implementations.
They also highlight a computational complexity separation between minimum-weight decoding and its approximate realizations.

More broadly, our work constitutes an important addition to the growing body of hardness results related to basic tasks in QEC and fault tolerance, such as computing the code distance~\cite{Berlekamp1978,Kapshikar2023}, decoding~\cite{Hsieh2011,Iyer2015,Kuo2020,khesin2025averagecasecomplexityquantumstabilizer}, and implementing or optimizing logical operators~\cite{Herr2017,Bostanci2021,Wasa2023}.
First, it recovers the very recent NP-hardness result of Walters and Turner for minimum-weight decoding of the color code~\cite{walters2026CCNPhard}, arguably through a more natural reduction.
Second, it establishes NP-hardness of minimum-weight decoding of the surface code, the prototypical and most studied QEC code.
This should be contrasted with the earlier result by Fischer and Miyake~\cite{Fischer2024SCNPhard}, which showed computational hardness of most-likely-error decoding for a fine-tuned Pauli noise with spatially inhomogeneous error rates.
Third, it demonstrates that the decoding problem for logical circuits, even in the very basic setting of the surface code with a transversal CNOT gate, is computationally challenging for minimum-weight decoding, motivating further research into efficient decoders for logical circuits~\cite{cain2025fastcorrelateddecoding,SerraPeralta2026}. 

\subsection{Future directions and open problems}
There are several interesting directions for future work. We proved hardness of decoding the color code and surface code on the most commonly used hexagonal and square lattices, respectively.
We expect similar results for other lattices or even non-uniform tilings of the plane, but the gadgets need to be designed carefully to avoid pathological configurations.
One may also consider decoding other code families, such as various quantum low-density parity-check codes~\cite{Breuckmann2021}, although our techniques do not straightforwardly carry over due to the lack of string-like operators.
These codes can be analyzed in the code-capacity setting or phenomenological noise setting when implementing a logical operation.

Given previous hardness results on various maximum-likelihood decoding problem for QEC codes~\cite{Iyer2015,Fischer2024SCNPhard}, one may ask if the degenerate decoding problems in the three scenarios we consider are \#P-hard. Problems in the complexity class \#P require counting the number of solutions to decision problems in NP. Counting the number of perfect matchings in \tdm is \#P-complete~\cite{GareyJohnson90NPcompleteness,Jerrum2003}, which roughly corresponds to finding the number of errors of minimal weight $w$ in our reduction. One way to complete the proof would be to bound the number of higher-weight errors multiplied by their probabilities in each equivalence class.

Another open question concerns approximating the minimum-weight solution to the decoding problems.
We show in Appendix~\ref{app_minwt} that a recovery whose weight is within a constant factor (two or three) of the minimum-weight recovery can be efficiently found.
However, it is unknown if an algorithm exists with approximation factor arbitrarily close to one, i.e., a polynomial-time approximation scheme.
If not, one may also ask if there is some $\varepsilon>0$ such that finding a solution with approximation factor $1+\varepsilon$ is NP-hard.

\begin{acknowledgments}
S.G. and A.K. thank Yue Wu and Cole Maurer for inspiring discussions about hardness of the color code decoding problem.
L.W. is supported by NSERC Fellowship PGSD3-587672-2024 and thanks the Yale Quantum Institute for hosting
her during the project.
S.G. and A.K. acknowledge support from the NSF (QLCI, Award No. OMA-2120757), IARPA and the Army
Research Office (ELQ Program, Cooperative Agreement No. W911NF-23-2-0219).

\end{acknowledgments}

\appendix

\section{Approximations to minimum-weight decoding}
\label{app_minwt}

We show that in each of the three decoding scenarios, an efficient, matching-based algorithm can find an error of weight within a constant factor of the minimum.

\subsection{Color code}
For the color code, we show that the restriction decoder from Ref.~\cite{Kubica2023CCrestrictiondecoder} returns an error of weight at most three times the minimum. Let $\sigma=\sigma_R\sqcup\sigma_G\sqcup\sigma_B$ be an $X$ syndrome decomposed as the three different defect types. For the three different colors $A, B, C$, the restricted lattice $\mc L_{AB}$ is defined by removing the $C$-colored vertices of the triangular lattice $\mc L$ and all adjacent edges.
The decoder first performs minimum-weight perfect matching (MWPM) on the three restricted lattices $\mc L_{AB}$ with syndrome $\sigma_A\sqcup\sigma_B$, obtaining the sets of edges $\gamma_{AB}$.
Without loss of generality, we may assume that $|\gamma_{RG}|\le|\gamma_{RB}|\le |\gamma_{GB}|$.
A local lifting procedure is then applied. In this step, for each red vertex $v$, a minimal subset $\tau_v$ of the adjacent faces is found such that its boundary is equal to $\gamma_{RG}\sqcup\gamma_{RB}$ when restricting to the edges incident to $v$. It was shown in Ref.~\cite{Kubica2023CCrestrictiondecoder} that the $Z$ operator $F$ with support $\bigsqcup_{v \text{ red}} \tau_v$ is a valid correction for the syndrome $\sigma$.

To bound the weight of $F$, let $E$ be a minimum-weight error with syndrome $\sigma$. The $\bb Z_2$-boundary of its support is a set of edges $\tilde\gamma$ of size at most $3|E|$, since each triangle has three edges.
We decompose $\tilde\gamma=\tilde\gamma_{RG}\sqcup\tilde\gamma_{RB}\sqcup\tilde\gamma_{GB}$, where $\tilde\gamma_{AB}$ consists of the edges connecting an $A$ vertex to a $B$ vertex for $A,B\in\{R,G,B\}$ distinct. Each $\tilde\gamma_{AB}$ is a matching of $\sigma_A\sqcup\sigma_B$ on the restricted lattice $\mc L_{AB}$. Since $\gamma_{AB}$ is a minimum-weight matching, we have $|\gamma_{AB}|\le |\tilde\gamma_{AB}|$.
Around every red vertex $v$ in the triangular lattice, the minimal local lift $\tau_v$ has size at most $3/2$ times the number of edges in $\gamma_{RG}\sqcup\gamma_{RB}$ incident to $v$. Summing over all red vertices gives
\begin{equation}
    |F|\le \frac{3}{2}(|\gamma_{RG}|+|\gamma_{RB}|) \le |\gamma_{RG}|+|\gamma_{RB}|+|\gamma_{GB}|\le |\tilde\gamma_{RG}|+|\tilde\gamma_{RB}|+|\tilde\gamma_{GB}|\le 3|E|.
\end{equation}

\subsection{Surface code}
Consider decoding the surface code with syndrome $\sigma$, and let $E=E_XE_Z$ be the minimum-weight error, where $E_X$ and $E_Z$ are the $X$ and $Z$ components of $E$, respectively. By using MWPM on the $X$ defects of $\sigma$, we find a minimum-weight $X$ error $F_X$ giving those defects. Similarly, MWPM gives the minimum-weight $Z$ error $F_Z$ that satisfies the $Z$ defects of $\sigma$. The combined error $F=F_XF_Z$ satisfies the syndrome and has weight
\begin{equation}
    |F| \le |F_X|+|F_Z|\le |E_X|+|E_Z|\le 2|E|.
\end{equation}

\subsection{Transversal CNOT decoding}
Now, let us consider decoding the surface code with transversal logical CNOT applied.
Let $\sigma$ be the syndrome, consisting of the set of triggered $X$ detectors. The syndrome can be expressed as the disjoint union $\sigma=\sigma_1\sqcup \sigma_2$, where $\sigma_1$ and $\sigma_2$ are the triggered detectors associated with the two surface codes.
(The detectors $s^1_{t_{G}-1/2}s^1_{t_G+1/2}s^2_{t_G-1/2}$ are considered to be in the first surface code.)
Let $E=E_1\sqcup E_2$ be the minimum-weight error decomposed as the errors affecting the each surface code.

We consider decoding the two surface codes sequentially, as was done in Refs.~\cite{Beverland21costofuniversality,wan2025iterativetcnotdecoder,Sahay25tCNOT,cain2025fastcorrelateddecoding}. Specifically, we first decode the second surface code using MWPM on the syndrome $\sigma_2$. This gives a set of $Z$ qubit errors and $X$ measurement errors $F_2$ of minimal weight that satisfies the detectors in $\sigma_2$. The error $F_2$ will also trigger certain detectors $\sigma_1^F$ in the first surface code patch if it contains measurement errors at time $t_G-\frac{1}{2}$. We then use MWPM to decode the syndrome $\sigma_1\oplus\sigma_1^F$ (the symmetric difference of $\sigma_1$ and $\sigma_1'$) within the first surface code to obtain a minimum-weight error $F_2$. Because errors on the first surface code do not trigger detectors associated with the second, the combined error $F=F_1\sqcup F_2$ satisfies the syndrome $\sigma$.

Let us now bound the weight of $F$. Because $E_2$ also satisfies the syndrome $\sigma_2$, we have $|F_2|\le |E_2|$.
Let $\sigma_1^E$ be the detectors triggered on the first surface code due to $E_2$, which consists defects at time $t_G$ at the spatial locations of measurement errors in $\sigma_1^E$ at time $t_G-\frac{1}{2}$. The syndrome of $E_1$ is $\sigma_1\oplus \sigma_1^E$.
Now, we define the $E^+$ and $E^-$ as the restrictions of $E_2\oplus F_2$ to the errors at times $t\ge t_G$ and $t\le t_G-\frac{1}{2}$, respectively, but on the first surface code. One of these operators has weight at most $\frac{1}{2}(|E_2|+|F_2|)$ since their union is $E_2\oplus F_2$.
Because $E_2$ and $F_2$ have the same syndrome when restricted to the second surface code and the two surface codes have the same time boundary conditions, the syndromes of $E^+$ and $E^-$ are both equal to $\sigma_1^E \oplus \sigma_1^F$. This means that $E_1 \oplus E^+$ and $E_1 \oplus E^-$ both have the syndrome $\sigma_1\oplus \sigma_1^F$. 
Because $F_1$ is the minimum-weight error with syndrome $\sigma_1\oplus \sigma_1^F$, we obtain
\begin{equation}
    |F_1|\le |E_1|+\frac{1}{2}(|E_2|+|F_2|)\le |E_1|+|E_2|,
\end{equation}
from which it follows that
\begin{equation}
    |F|\le |F_1|+|F_2|\le 2|E_1|+|E_2|\le 2|E|.
\end{equation}

\bibliography{refs}

\end{document}